\documentclass{article}
\usepackage[a4paper]{geometry}
\usepackage[utf8]{inputenc}
\usepackage{amsfonts} 
\usepackage{amsmath}
\usepackage{systeme}
\usepackage{setspace}
\usepackage{cite}
\usepackage{booktabs}
\usepackage{setspace}
\usepackage{array}
\usepackage[utf8]{inputenc}
\usepackage{fancyhdr}
\usepackage{graphicx}
\usepackage{mathtools}
\usepackage{pgfplots}
\pgfplotsset{compat=1.15}
\usepackage{mathrsfs}

\usepackage{pgf,tikz}
\usetikzlibrary{arrows}
\usepackage{cases}
\usepackage{amssymb}
\usepackage{amsthm}
\usepackage{tabularx}
\usepackage{chessboard}
\usepackage{tikz-network}
\usepackage{tikz}
\usepackage{bbm}

\usepackage{mathtools}

\usetikzlibrary{shapes,backgrounds}
\usetikzlibrary{decorations.markings}
\usepackage{verbatim}
\makeatother
\usepackage[inline]{asymptote}
\usepackage{extarrows}
\usepackage{hyperref}
\hypersetup{
    colorlinks,
    citecolor=black,
    filecolor=black,
    linkcolor=black,
    urlcolor=black
}

\newcolumntype{C}[1]{>{\centering\arraybackslash}m{#1}}

\newtheorem{theorem}{Theorem}

\title{Cramming and Credibility: Strategic Test Announcements in the Classroom}
\author{Zijun Meng}
\date{}

\begin{document}
\linespread{1.5}
\setlength{\baselineskip}{18pt}
\maketitle

\begin{center}
    \textbf{Abstract}
\end{center}
This paper studies a cheap-talk model of strategic test announcements. A teacher observes the day of the test of the next week decided by the nature and makes an announcement to his students who choose effort levels of studying. The competing forces are the teacher's value on consistent study habits and the students' grade orientation. We characterize the pure strategy Nash equilibrium under the linear-quadratic student utility. We also study what happens when the teacher can commit to an information policy.

\tableofcontents

\newpage
\section{Introduction}
There is a teacher and a group of students. The teacher prepares weekly tests for the students, but the day of the test is random, it could be any day from Monday to Friday. The teacher would like the students to develop a good habit of studying consistently regardless of whether there is an upcoming test because he believes that this will be beneficial for the students in the long run, but students can be too lazy to do that, and since they still care about their grades, they have the incentive to study more when they guess that there will be an upcoming test tomorrow. There is a spectrum of students in this regard: some focus more on their long-term development, some others are more short-sighted and only care about the grades.

In the end of every week, the teacher will claim that the test for the next week will be on a certain day. He could tell the truth and he could lie. In order to let the students to study not only right before the tests, the teacher should not tell the truth every time. On the other hand, in order to let the students not totally ignore what he says (or even directly regard the days he announce as the days with no tests), the teacher should also not lie every time. From the perspective of the students, they should also be strategic on updating their beliefs on the days of the tests upon hearing the announcements of the teacher.

In this paper, we study how the teacher and the students interact in the aforementioned setting by setting up a mathematical model. We characterize the Nash equilibria of this strategic setting. We also try to understand how the preferences (i.e. the far-sightedness) of the students affect the outcomes (for example, the extent to which the students distrust the teacher, and the way the students study) in the equilibria. In addition, we study what happens when the teacher can commit to a strategy before the truth test days are decided. 

The idea is rooted in the cheap-talk model in \cite{c1982} and \cite{f1996}, which is later extended to a setting of multiple audiences in \cite{f1989}, \cite{g2008} and a dynamic setting in \cite{a2003}, \cite{g2014}. Also, \cite{k2011} provides a useful framework of Bayesian persuasion, on which our result on the commitment setting is based.

\section{Model}
Let $D=\{1,\,2,\,3,\,4,\,5\}$ be the set of days in a week, denoting Monday, Tuesday, Wednesday, Thursday and Friday respectively. For a typical week, the teacher stochastically decides the day $\theta\in D$ of the test with uniform probability on each day, i.e. $\pi_i:=\mathbb P(\theta=i)=\frac15$ for any $i\in D$. The teacher then tells the students that the test will be on day $m\in D$ according to his/her strategy $\sigma(m|\theta)$, where \[\sum_{m\in D}\sigma(m|\theta)=1.\]

Students differ in terms of far-sightedness, which is captured by his/her type $\alpha\in [0,\,1]$. Suppose that the types of the class of students follow a distribution $F$. A student with type $\alpha$ enjoys a utility \[u_\alpha(x_\alpha,\theta)=\alpha\sum_{d\in D}H(x_{\alpha d})+(1-\alpha)G(x_{\alpha\theta})-\sum_{d\in D}C(x_{\alpha d}),\] where $x_\alpha=(x_{\alpha 1},\,\dots,\,x_{\alpha 5})$ is the vector denoting the effort the student spend on studying for each day, $H$ is the part of utility gained by maintaining a good habit of student consistently, $G$ is the part of utility gained by improving the grade, and $C$ is the cost of studying.

The students have partial distrust on the message $m$ sent by the teacher, but for the moment let us suppose that the students have a consensus and they share a posterior belief that the probability that the true test day is on day $d\in D$ upon receiving message $m$ is \[\mu(d| m):=\mu_{(\pi,\sigma)}(d| m):=\frac{\pi_d\sigma(m|d)}{\sum\limits_{\theta\in D}\pi_\theta\sigma(m|\theta)}=\frac{\sigma(m|d)}{\sum\limits_{\theta\in D}\sigma(m|\theta)},\] and \[\sum_{d\in D}\mu(d| m)=1.\]

A student of type $\alpha$ now tries to solve for \[\max_{x_{\alpha d}\ge 0}\sum_{d\in D}(\alpha H(x_{\alpha d})+(1-\alpha)\mu(d| m)G(x_{\alpha d})-C(x_{\alpha d})).\] Assuming sufficient smoothness of the functions we can write the first-order condition as \[C'(x_{\alpha d})=\alpha H'(x_{\alpha d})+(1-\alpha)\mu(d| m)G'(x_{\alpha d}).\]

The teacher hopes that students spend an effort level $x^T_d(\theta)=a+\tau\delta_{d\theta}$, where $\tau$ is a parameter capturing the degree to which the teacher hopes students to focus on the test (i.e. if the teacher hopes students to completely study for its own sake and ignore the test, then $\tau=0$). The teacher enjoys a negative $\mathcal L^2$ utility \[U_T(\theta,\,m)=-\int_{[0,\,1]}\sum_{d\in D}(x^*_{\alpha d}(m)-x^T_d(\theta))^2dF,\] where $x^*_{\alpha d}(m)$ is the effort level on day $d$ for the utility of a student with type $\alpha$ upon hearing message $m$.

Now, the model is not very tractable as H, G and C are all generic. Let's suppose that $H(x_{\alpha d})=hx_{\alpha d}$, $G(x_{\alpha\theta})=gx_{\alpha\theta}$, $C(x_{\alpha d})=\frac c2x_{\alpha d}^2$. The first order condition \[C'(x_{\alpha d})=\alpha H'(x_{\alpha d})+(1-\alpha)\mu(d\mid m)G'(x_{\alpha d})\] provides the optimal effort level \[x^*_{\alpha d}=\frac{\alpha h+(1-\alpha)g\mu(d\mid m)}c.\]

Then, the teacher's utility $\tilde U_T(\theta,\,m)$ becomes \begin{align*}&-\int_{[0,\,1]}\sum_{d\in D}(x^*_{\alpha d}(m)-x^T_d(\theta))^2dF\\=&-\int_{[0,\,1]}\sum_{d\in D}\left(\frac{\alpha h+(1-\alpha)g\mu(d| m)}c-x_d^T(\theta)\right)^2dF\\=&-\int_{[0,\,1]}\left(\sum_{d\in D}\frac{(1-\alpha)^2g^2\mu(d| m)^2}{c^2}-\frac{2\tau (1-\alpha)g\mu(\theta| m)}c+\text{constant}\right)dF\\=&-\frac{g^2}{c^2}\left(\int_{[0,\,1]}(1-\alpha)^2dF\right)\sum_{d\in D}\mu(d| m)^2+\frac{2\tau g}c\left(\int_{[0,\,1]}(1-\alpha)dF\right)\mu(\theta| m)+\text{constant}\\=&\frac{g^2}{c^2}\left(\int_{[0,\,1]}(1-\alpha)^2dF\right)\left(-\sum_{d\in D}\mu(d| m)^2+2\kappa\mu(\theta| m)\right)+\text{constant},\end{align*} \[\text{where }\kappa:=\frac{\frac{\tau g}c\int_{[0,\,1]}(1-\alpha)dF}{\frac{g^2}{c^2}\int_{[0,\,1]}(1-\alpha)^2dF}.\]

\section{Pure Equilibria}
Up to a increasing (in fact, affine) transformation, it suffices to consider \[U_T(\theta,\,m):=-\sum_{d\in D}\mu(d| m)^2+2\kappa\mu(\theta| m)\] (instead of $\tilde U_T$) as the teacher's utility.
A teacher's pure strategy equilibria can be written as a function $s:D\to D$, where the teacher tells the students that the test will be on day $s(\theta)$ when the true test day is $\theta$.
$s$ induces fibers $C_m=\{\theta\in D:s(\theta)=m\}$, and hence a partition $\{C_1,\,\dots,\,C_K\}$ of $D$, where $C_m$ is the set of true test days for which the teacher will send a message $m$.
Students will have a Bayesian update $\mu(d| m)=\frac1{|C_m|}\mathbbm 1_{d\in C_m}$ upon hearing message $m$.

\begin{theorem}
The necessary and sufficient condition for \[\left(s,\,x^*_{\alpha d}=\frac{\alpha h+(1-\alpha)g\frac1{|C_m|}\mathbbm 1_{d\in C_m}}c\right)\] to be a Nash equilibrium is \[\kappa\ge\frac{\max\limits_j|C_j|-\min\limits_j|C_j|}{2\max\limits_j|C_j|}.\]
\end{theorem}

\begin{proof}
We have argued before that $x^*_{\alpha d}$ is the best response for a student of type $\alpha$, as for the teacher not to deviate, we need $U_T(\theta,\,m_j)\ge U_T(\theta,\,m_k)$ for any $j$, $k$, and for any $\theta\in C_j$, and we have \[U_T(\theta,\,m_j):=-\sum_{d\in D}\mu(d| m_j)^2+2\kappa\mu(\theta| m_j)=-\frac1{|C_j|}+2\kappa\frac1{|C_j|}\mathbbm 1_{\theta\in C_j}=\frac{2\kappa-1}{|C_j|}\] and $U_T(\theta,\,m_k)=-\frac1{|C_k|}$ for $\theta\not\in C_k$, so the inequality reduces to \[\frac{1-2\kappa}{|C_j|}\le\frac1{|C_k|},\] rearranging gives the aforementioned bound on $\kappa$.
\end{proof}

We have the following table for $\kappa_{\text{min}}$ to support equilibria for various induced partitions:

\begin{table}[h]
    \centering
    \begin{tabular}{|C{2.7cm}|C{0.7cm}|}
        \hline
        Cell-size Profile & $\kappa_{\text{min}}$\\
        \hline
        $5$ & $0$ \\
        \hline
        $4+1$ & $\frac38$ \\
        \hline 
        $3+2$ & $\frac16$ \\
        \hline
        $3+1+1$ & $\frac13$ \\
        \hline 
        $2+2+1$ & $\frac14$ \\
        \hline
        $2+1+1+1$ & $\frac14$ \\
        \hline
        $1+1+1+1+1$ & $0$ \\
        \hline
    \end{tabular}
\end{table}

The intuition is that more imbalanced partitions require a higher teacher preference for test-day preparation, due to the teacher's incentive compatibility. For example, the optimality of the profile $4+1$ requires a high $\kappa$ to support, because the teacher would be largely incentivized to deviate when he meets the singleton type (because otherwise the students knows the true test day for sure and thus cram), whereas for the profile $3+2$ the requirement for $\kappa$ is lower, as the probability $\frac13$ and $\frac12$ for the two types are not as far away as compared to $\frac14$ and $1$. The profile $1+1+1+1+1$ and $5$ are uniform, so there is no requirement whatsoever for the teacher's preference on cramming.

\section{Commitment}
In this section, we study what happens when the teacher can commit to a messaging strategy $\sigma$ before the realization of the true test day $\theta$. The students' best responses $x^*$ in the previous theorem carry over, it suffices for us to consider what the teacher does.

\begin{theorem}
When the teacher can commit, the teacher-optimal information policy is \[\begin{cases}\text{babbling},\,&\kappa<\frac12\\\text{anything},\,&\kappa=\frac12\\\text{full revelation},\,&\kappa>\frac12\end{cases},\] where babbling means $\sigma(m|\theta)=\frac15$ for any $m$, $\theta$, full revelation means that there is some function $s:D\to D$ such that $\sigma(s(\theta)|\theta)=1$ and $\sigma(d|\theta)=0$ for any $d\ne s(\theta)$.
\end{theorem}

\begin{proof}
Upon sending a message $m$, the expected utility of the teacher is \begin{align*}\sum_{\theta\in D}\mu(\theta|m)U_T(\theta,\,m)&=\sum_{\theta\in D}\mu(\theta|m)\left(-\sum_{d\in D}\mu(d|m)^2+2\kappa\mu(\theta|m)\right)\\&=2\kappa\sum_{\theta\in D}\mu(\theta|m)^2-\sum_{d\in D}\mu(d|m)^2\\&=(2\kappa-1)\sum_{\theta\in D}\mu(\theta|m)^2.\end{align*} Hence, if $\kappa<\frac12$, the teacher minimizes $\sum\limits_{\theta\in D}\mu(\theta|m)^2$, so he wants to choose a $\sigma$ so as to make the posterior distribution uniform. If $\kappa>\frac12$, the teacher maximizes $\sum\limits_{\theta\in D}\mu(\theta|m)^2$, so he wants to choose a $\sigma$ so as to make the posterior distribution concentrated. If $\kappa=\frac12$, the teacher does not care.
\end{proof}

\section{Possible Extensions}
Some possible future research directions include studying what happens during the week, meaning how the students would update their beliefs intra-week. Another important question is what happens in the second week, the third week, when the teacher and the students can learn about each other's behavior. We can also study the impact of peer effects, in other words, what happens when the students care about how others look at themselves, beyond their own grades and long term development?

\end{document}